\documentclass[12pt]{article}

\usepackage{
amsmath,
amsfonts,
amssymb,
amscd,
amsthm, 
enumerate, 
}

\usepackage{dsfont} 

\usepackage{geometry}
\usepackage{graphicx}

\newtheorem{theorem}{Theorem}[section]
\newtheorem{proposition}[theorem]{Proposition}
\newtheorem{lemma}[theorem]{Lemma}
\newtheorem{corollary}[theorem]{Corollary}

\theoremstyle{definition}

\renewcommand{\leq}{\leqslant}
\renewcommand{\geq}{\geqslant}
\newcommand{\la}{\langle}
\newcommand{\ra}{\rangle}

\newcommand{\da}{\downarrow}
\newcommand{\ua}{\uparrow}
\newcommand{\uap}{\uparrow^+}

\newcommand{\yu}{y\uparrow}

\newcommand{\allnat}{[\omega]}
\newcommand{\finnat}{\omega^{<\omega}}
\newcommand{\finbin}{2^{<\omega}}
\newcommand{\rar}{\rightarrow}
\newcommand{\computes}{\longrightarrow}
\newcommand{\compeq}{\longleftrightarrow}
\newcommand{\ovl}[1]{\overline{#1}}

\newcommand{\zetZero}{z^{4i}}
\newcommand{\zetOne}{z^{4i + 1}}
\newcommand{\zetTwo}{z^{4i + 2}}
\newcommand{\zetThree}{z^{4i + 3}}

\newenvironment{acknowledgement}{\noindent\bf Acknowledgment\rm}{}

\newenvironment{keywords}{\noindent\sc Keywords:\rm}{\mbox{}}

\begin{document}

\title{Additivity of on-line decision complexity is violated by a linear term in the length of a binary string DRAFT}

\author{Bruno Bauwens
\thanks{
The result and its motivation was presented at 2009 conference of Logic, Computability
and Randomness in Luminy \cite{LuminyTalk}.
Department of Electrical Energy, 
Systems and Automation, Ghent University, 
Technologiepark 913, B-9052, Ghent, Belgium, 
Bruno.Bauwens@ugent.be.
Supported by a Ph.D grant of the Institute for the Promotion of 
Innovation through Science and Technology in Flanders (IWT-Vlaanderen).} 
}

\date{\normalsize \today}

\maketitle  

\begin{abstract}
 We show that there are infinitely many binary strings $z$, such that the sum 
of the on-line decision complexity of predicting the even bits of $z$ given the 
previous uneven bits, and the decision complexity of predicting the uneven bits 
given the previous event bits, exceeds the Kolmogorov complexity of $z$ by a 
linear term in the length of $z$.
 \\[5pt]
\end{abstract}

\section{Introduction} 

On-line decision complexity has been introduced and investigated in \cite{ShenRelations,
onlineComplexity}.  It also naturally appears 
in the definition of ideal influence tests \cite{AIT, LuminyTalk}.
A natural question is whether algorithmic mutual information of two time series $x,y$, can be 
decomposed into an information flow going from $x$ to $y$, a flow going from
$y$ to $x$, and an information flow instantaneously present in both strings. 
It turns out \cite{AIT} that this question is related to the question of 
defining a decomposition of $K(x,y)$ with $l(x)=l(y)$ as the sum of the complexity of predicting 
$x_{i+1}$ given $x_1...x_i$ and $y_1...y_i$, $i\leq n$, and the complexity of predicting 
$y_{i+1}$ given $x_1...x_{i+1}$ and $y_1...y_i$.  It will be shown that using on-line decision 
complexity for this complexity, this sum exceeds $K(x,y)$ by 
a linear constant in $l(x)$.
A modification of this definition of on-line decision complexity will be shown 
to have an approximate decomposition \cite{AIT, LuminyTalk}.  

Non-additivity of decision complexity was also shown in \cite{Muchnik}, 
in the context of randomness defined by supermartingales. Using natural definitions
for randomness a paradox is shown: if 
the even bits of $z$ given the past uneven bits of $z$ are random, 
and also the uneven bits of $z$ given the past even bits of $z$ are random, 
than it is possible that $z$ is not random. The proof of this result implies that 
additivity of on-line decision complexity is violated by a logarithmic term.

\section{Definitions and notation} \label{sec:defs}

For excellent introductions to Kolmogorov complexity we refer to 
\cite{GacsNotes, LiVitanyi}.
Let $\omega$, $\finnat$, $2^N$ and $2^{<\omega}$ denote the set of the Natural numbers, 
the set of finite sequences of Natural numbers, the binary strings of length $N$, 
and the binary strings of finite length. Other definitions are analogue. Let $\epsilon$
denote the empty sequence. Remark that 
there is a natural bijection between $\omega$ and $2^{<\omega}$, defined by: 
$$
\epsilon \rar 0, 0 \rar 1, 1 \rar 2, 00 \rar 3, 01 \rar 4, ... 
$$
$\allnat$ is the set of nested sequences of Natural numbers, with finite depth.
Mathematically, it is the closure of $\omega$ under the mapping $f(S) = S^{<\omega}$. 
Remark that there is a computable bijection between $\omega$ and $\allnat$, therefore
most complexity and computability results in $\omega$ also hold in $\allnat$.

An interpreter $\Phi$ is a partial computable function from $\finbin \times \allnat \rar \allnat$.
An interpreter is prefix-free if for any $x$, the set $D_x$ of all $p$ where 
$\Phi(p|x)$ is defined, is prefix-free.
Let $\Phi$ be some fixed optimal universal prefix-free interpreter. 

For any $x \in \finbin$, $l(x)$ denotes the length of $x$. For any $x \in \finnat$, 
$l(\overline{x})$ corresponds to the length of some prefix-free encoding of $x$ on a 
binary tape:
$$
 l(\overline{x}) = \sum_{i=1}^{l(x)} 2\log x_i.
$$

\noindent
For $x,y \in \allnat$, the \textit{Kolmogorov complexity} $K(x|y)$, is defined as:
$$
 K(x|y) = \min\{l(p): \Phi(p|y) \da=x\}.
$$
The Kolmogorov complexity of elements in $\finbin$ is defined by using
the computable bijection mentioned in the beginning of this section.

\noindent 
For $Z \in \allnat, Q,A \in \omega^n$, $Q^i$ denotes $Q_1...Q_i$. The on-line decision complexity is defined by:
$$
 K(Q_1 \rar A_1; ... ; Q_n \rar A_n|Z) = \min\{l(p): \forall i<n [\Phi(p|Q^i,Z) \da=A_i]\}.
$$
This definition differs slightly with the definition of \cite{onlineComplexity}, 
with respect that $A \in \omega^n$ is chosen, in stead of $A \in 2^n$. 
Also a shorter notation \cite{AIT} will be used:
\begin{eqnarray*} \label{eq:Kup}
  K(x|y\ua) &=& K(0 \rar x_1; ... ; y_{n-1} \rar x_n), \\
  K(y|x\uap) &=& K(x_1 \rar y_1; ... ; x_n \rar y_n). 
\end{eqnarray*}

\section{Main result and proof tactic}

\begin{proposition} \label{prop:decompMain}
$$
\exists c>0\exists^\infty x,y \in \finnat \big[ K(x|\yu) + K(y|x\uap) - K(x,y) > c(l(\overline{x}) + l(\overline{y})) \big].
$$
\end{proposition}

\noindent
In $\cite{ComplexityOfComplexity}$ and repeated in $\cite{GacsNotes, LiVitanyi}$, 
it is proven that for any $n$ there is an $x \in 2^n$ such that:
$$
 K(K(x)|x) \geq^+ \log n - \log \log n.
$$
Let $y$ be the binary expansion of $K(x)$. From this and 
equation \eqref{eq:complexityOfComplexity} it can be shown that 
$$
 K(x) + K(y|x) - K(x,y) \geq^+ \log n - \log \log n .
$$
By inserting zeros at the right places in $x,y$, it can be shown that there exists 
infinity many $x,y$ with $l(x) = l(y)$:
$$
K(x|\yu) + K(y|x\uap) - K(x,y) > O(\log l(x)).
$$
This shows proposition \ref{prop:decompMain} for a logarithmic term in $l(x)$. 
It seems natural to think that such a result can be improved to a linear term, 
by concatenating such strings. This is what eventually will happen in the proof,
 at equation \eqref{eq:KTyLarge}.
 However, to be able to add up these differences, 
 conditional complexities must add up in some way to 
on-line decision complexity, in what extend this is possible is still an open 
problem.  Happily, 
Lemma \ref{lem:decisionComplexityLB} can circumvent this, 
if some extra information is available. 
This information is stored in sequences 
$u$ and $v$ and is added to 
$x$ and $y$. Adding this information 
requires, some more bounds to make the proof work: \eqref{eq:puttingOutY},
\eqref{eq:puttingInY}. The proof below provides all technical details.

\section{Proof}

\noindent
First some definitions and lemmas are given. 
$f(x) \leq^+ g(x)$ is short for $f(x) \leq g(x) + O(1)$, 
and $f(x) =^+ g(x)$ is short for $f(x) = g(x) \pm O(1)$.
For any $a,b \in \allnat$, $a \computes b$ means that 
there is a fixed $p \in \finbin$
with $l(p) \leq O(1)$, such that $\Phi(p|a) \da= b$.
Remark that if $a \computes b$, then $K(a) \geq^+ K(b)$.
The shortest program witnessing $K(a|b)$ is denoted by:
$$
a^*[b] = \min \{p : \Phi(p|b) \da= a\}.
$$
$a^*$ is short for $a^*[\epsilon]$. Remark that: 
\begin{equation} \label{eq:shortestOfShortest}
 a^*[b],b \compeq (a^*[b])^*[b],b.
\end{equation}

\noindent
Lemmas \ref{lem:sequentialAdditivity},
\ref{lem:complexityOfComplexity}, and \ref{lem:KshortestUB} provide 
observations, known within the community, 
and stated here explicitly for later reference.

\begin{lemma} \label{lem:sequentialAdditivity}
 For $A \in \omega^{<\omega}$, 
 $$
  \sum_{i\leq n} K(A_i|A^{i-1}) \geq K(A) - O(n).
 $$
\end{lemma}

\begin{proof}
For $U,V \in \omega$, prefix-free complexity satisfies additivity \cite{LiVitanyi}: 
\begin{equation} \label{eq:Kadditivity}
 K(U,V|W) =^+ K(U|W) + K(V|U^*[W]). 
\end{equation}
Since there is a computable bijection between $\omega$ and $\allnat$, 
this result also applies to $\allnat$. 
Let $U,V \in \allnat$, since $U^*[W],W \computes U$, 
$$
K(V|U,W) \geq^+ K(V|U^*[W],W).
$$ 
Inductive application of both equations above on $A^i$ proves the lemma.
\end{proof}

\begin{lemma} \label{lem:complexityOfComplexity}
 For $a,b \in \omega$ and $c \in \finnat$:
 $$
  K(a,b|c) =^+ K(a,b, K(b|a^*[c],c)|c).
 $$
\end{lemma}

\begin{proof}
The proof below, shows the unconditioned version of the lemma, 
since the proof of the conditioned version is the same. 
In \cite{GacsNotes} and exercise $3.3.7$ in \cite{LiVitanyi} 
it is stated that for every $w \in \omega$, and $n \geq K(w)$:
\begin{equation} \label{eq:NoShortestPrograms}
  \log |\{p \in 2^n: \Phi(p) \da= w\}| \leq^+ n - K(w,n),
\end{equation}
and
\begin{equation} \label{eq:complexityOfComplexity}
 K(w,K(w)) =^+ K(w).
\end{equation}
Therefore, for $c$ constant, there are an $O(1)$ number of programs that 
compute $a,b$ and have length $K(a,b) + c$.
Let $S$ be the set of these programs.
Remark that the elements of $S$ can be enumerated given $a,b,K(a,b)$ and therefore, 
for any $p \in S$, using \eqref{eq:complexityOfComplexity}, we have: 
\begin{equation} \label{eq:KabVsKp}
 K(a,b) =^+ K(a,b,K(a,b)) =^+ K(p).
\end{equation}
By equation \eqref{eq:Kadditivity}, we have:
$$
 K(a,b) =^+ K(a) + K(b|a^*).
$$
The programs $a^*$ and $b^*[a^*]$, can be combined into  
a program $p$ computing $a,b$. This program $p$ can be constructed 
such that $p \computes a^*, b^*[a^*]$, and it has a length below 
$K(a,b) + c$, for $c$ constant and large enough.
Therefore $p \in S$, and since  
$b^*[a^*] \computes K(b|a^*) = l(b^*[a^*])$:
$$
K(p) \geq^+ K(a,b, K(b|a^*)).
$$
Combining with equation \eqref{eq:KabVsKp}, finishes the proof.
\end{proof}

\begin{lemma} \label{lem:KshortestUB}
 For $b \in \finbin$, $a,c \in \allnat$:
 $$
  K(a,b|c) \geq^+ K(a, b^*[a,c]|c) - 2\log l(b).
 $$
\end{lemma}

\begin{proof}
The unconditioned version of the lemma is proven, since the 
conditioned proof is essentially the same. It suffices to show that:
$$
 K(b^*[a] | a,b) \leq^+ 2\log l(b).
$$
Again the proof of the unconditioned version of this equation is the 
same as the conditioned one:
$$
 K(b^*|b) \leq^+ 2\log l(b).
$$
Given $b$ and $K(b)$ all programs of length $K(b)$ that output $b$ 
can be enumerated. By equation \eqref{eq:NoShortestPrograms}, 
there are maximally a constant such programs, therefore: 
$$
 K(b^*|b) =^+ K(K(b)|b).
$$
Remark that by the prefix-free code $b_10b_20...b_{l(b)}1$ we have:
$$
 K(b) \leq^+ 2 l(b).
$$
Using the natural bijection between $\omega$ and $2^{<\omega}$, 
this shows that for $n \in \omega$, $K(n) \leq^+ 2\log n$.
$$
K(K(b)|b) \leq^+ K(K(b)) \leq^+ 2\log K(b) \leq^+ 2\log l(b).
$$ 
\end{proof}

\noindent
Let $Z \in \allnat$, $A,Q \in \omega^n$ for some $n$, and $N \in \omega$.
For $i<n$, 
let $T_i=(Q_i | A_i)$ and $T = (T_1, ..., T_n)$. 
\begin{eqnarray*}
 K(T) 	&=& K(A|Q\ua,N) \\
 K(T_i|Z) &=& K(A_i|A^{i-1},Q^{i-1},N,Z).
\end{eqnarray*}
For some fixed $N$, and for all $i\leq n$, we define the sets $S_i$ and the numbers $L_i$:
\begin{eqnarray*}
 S_0(T) &=& 2^N \\
 S_i(T) &=& S_{i-1} \cap \{p: \Phi(p|Q^{i},N) \da=A_i \}\\
 L_i(T) &=& 
  \begin{cases} -1 			& \text{if } |S_i(T)| = 0 \\
    \lceil \log | S_i(T) | \rceil 	& \text{otherwise.} 
  \end{cases}
\end{eqnarray*}

\noindent A lower bound for $K(T)$ is now proven.
\begin{lemma} \label{lem:decisionComplexityLB}
 $$
   K(T) \geq \min\{N, \sum_i K(T_i|L_{i-1}) - O(n)\}.
 $$
\end{lemma}

\begin{proof}
 For each $i$, a semimeasure $P$ can be constructed using $A^{i-1}$,$Q^i$,$L_{i-1}$,$N$:
 $$
  P(z) = 2^{-L_{i-1}} |\{p \in S_{i-1} : \Phi(p|A^{i-1},Q^i,N) \da= z\}|.
 $$
 Remark that $P$ defines a semimeasure and that $P$ is enumerable.
 $P(A_i)=0$, for some $i$, implies that no program of length $N$ can solve task $T^i$, thus $K(T) > N$. 
In this case the lemma is proven.  Assume $|S_i| \geq 1$ and thus $P(A_i)>0$.  
By applying the coding theorem \cite{LiVitanyi} on $P$, it follows that: 
 $$
  L_{i-1} - L_i \geq K(T_i|L_{i-1}) - O(1).
 $$
 Summing over $i$, gives:
 \begin{equation} \label{eq:L_vs_KTi}
 L_0 - L_n \geq \sum_i K(T_i|L_{i-1}) - O(n).
 \end{equation}
 Let $p$ be a program of length $K(T)$, solving task $T$.  It possible to append 
 $2^{N - K(T) - O(1)}$ different strings of length $N - K(T) - O(1)$
 to $p$, in order to obtain elements from $S_n$. Therefore:
 \begin{equation} \label{eq:Ln_vs_KT}
   L_n \leq^+ N - K(T). 
 \end{equation}
 Observe that $L_0 = N$. Combining equations \eqref{eq:L_vs_KTi} and \eqref{eq:Ln_vs_KT} proves the lemma.
\end{proof}

\begin{proof} \textit{of proposition \ref{prop:decompMain}}.
Let $u,x,y,v \in \omega^n$ for some $n$. Let 
$$
z = N,0,0,0,u_1, x_1, y_1, v_1, ... , u_n, x_n, y_n, v_n.
$$

Define:
\begin{eqnarray*}
 T_{ux,i} &=& (u_i, x_i | z^{4i}) \\
 T_{x,i} &=& (x_i | z^{4i + 1}) \\
 T_{yv,i} &=& (y_i, v_i | z^{4i + 2}).
\end{eqnarray*}
For $X = ux, yv$, let $D_{X,1}=0$ and for $i\geq 2$ let:
\begin{eqnarray*}
 D_{X,i} &=& L_{i-1}(T_X) - L_i(T_X).
\end{eqnarray*}
Remark that:
$$
 \sum_{j\leq i} D_{X,j} = N - L_i(X).
$$

\noindent
Equations \eqref{eq:DKdecomp}, \eqref{eq:puttingInY}, 
\eqref{eq:puttingOutY}, and \eqref{eq:KTyLarge} are now derived.

\begin{itemize}

\item
Let: 
\begin{eqnarray*}
 u_i &=& D_{yv,i-1}^*[z^{4i}] \\
 v_i &=& D_{ux,i-1}.
\end{eqnarray*}

\noindent 
At the end of the proof $u,x,y,v,N$ will be constructed such that equation 
\eqref{eq:xyLengthB} holds, and therefore, 
$N \geq K(T_X) - O(n)$ for $X=ux,yv$. Since 
\begin{eqnarray*}
 \zetZero & \computes & u_i  \computes  L_{i-1}(T_{ux}) \\ 
 \zetTwo  & \computes & v_i  \computes  L_{i-1}(T_{yv})
\end{eqnarray*}
we have by lemma \ref{lem:decisionComplexityLB}: 
\begin{equation} \label{eq:DKdecomp}
  K(T_X) \geq \sum_i K(T_{X,i}) - O(n).
\end{equation}

\item
\noindent Choose:
\begin{equation} \label{eq:yDef}
 y_i = K(T_{x,i})^*[z^{4i+2}]. 
\end{equation}
By Lemma \ref{lem:complexityOfComplexity}, it follows that:
$$
K(u_i, x_i|z^{4i}) =^+ K(u_i, x_i, K(x_i | u_i^*[z^{4i}], z^{4i}) | z^{4i}).
$$
By equation \eqref{eq:shortestOfShortest}, we have that $u_i^*[\zetZero], \zetZero \compeq u_i, \zetZero$, and therefore: 
\begin{eqnarray*}
K(x_i | u_i^*[z^{4i}], z^{4i}) &\makebox[0.34cm][l]{$=^+$}& K(x_i|u_i, z^{4i}) \\
		&\makebox[0.34cm][l]{$=$}& K(x_i | z^{4i+1}) \\
		&\makebox[0.34cm][l]{$=$}& K(T_{x,i}) 
\end{eqnarray*}
Therefore:
$$
K(u_i, x_i, K(x_i | u_i^*[z^{4i}], z^{4i}) | z^{4i}) =^+ K(u_i, x_i, K(T_{x,i}) | z^{4i}). 
$$
Remark that $l(x_i) = m$, and therefore $K(T_{x,i}) \leq^+ 2\log m$.
 By Lemma \eqref{lem:KshortestUB}, we have:
$$
 K(u_i, x_i, K(T_{x,i}) | z^{4i}) \geq K(u_i, x_i, K(T_{x,i})^*[\zetTwo] | z^{4i}) - O(\log \log m).
$$
 By definition of $y_i$, \eqref{eq:yDef}, this shows that:
\begin{equation} \label{eq:puttingOutY}
 K(u_i, x_i|z^{4i}) \geq K(u_i, x_i, y_i | z^{4i}) - O(\log \log m).
\end{equation}

\item
From equations \eqref{eq:shortestOfShortest} and \eqref{eq:yDef}, we have: 
$$
 y_i,z^{4i+2} \compeq y_i^*[z^{4i+2}],z^{4i+2}.
$$
Therefore, 
\begin{eqnarray}
 K(v_i | z^{4i+3}) &\makebox[0.34cm][l]{$=$}& K(v_i | y_i, z^{4i+2}) \nonumber \\
		&\makebox[0.34cm][l]{$=^+$}& K(v_i | y_i^*[z^{4i+2}], z^{4i+2}) \nonumber \\ 
		&\makebox[0.34cm][l]{$=^+$}& K(y_i, v_i | z^{4i+2}) - K(y_i |z^{4i+2}) \label{eq:puttingInY}
\end{eqnarray}

\item
\noindent In \cite{ComplexityOfComplexity, GacsNotes, LiVitanyi} it is shown that 
for all $m, w$ there is an $x \in 2^m$ such that 
$$
 K(K(x|w)| x,w) \geq \log m - \log \log m - O(1).
$$
Actually, the unconditioned version is shown, but this version has the same proof.
Fix an $m$ large enough and choose $x_i \in 2^m$ such that by equation \eqref{eq:yDef}:
\begin{eqnarray}
 K(y_i | \zetTwo)  &\makebox[0.34cm][l]{$=$}& K(K(T_{x,i})^*[\zetTwo]|\zetTwo) \nonumber \\  
                  &\makebox[0.34cm][l]{$\geq^+$}& K(K(T_{x,i})|\zetTwo) \nonumber \\  
                  &\makebox[0.34cm][l]{$=$}& K(K(x_i|\zetOne)|x_i, \zetOne) \nonumber \\  
                  &\makebox[0.34cm][l]{$\geq^+$}& \log m - \log \log m.   \label{eq:KTyLarge}
\end{eqnarray}

\end{itemize}

\noindent
First using Lemma \ref{lem:sequentialAdditivity}, then applying subsequently 
equations \eqref{eq:puttingOutY}, \eqref{eq:puttingInY}, \eqref{eq:KTyLarge}, and 
\eqref{eq:DKdecomp} gives:

\begin{eqnarray*}
	&& K(u  ,x,y,v) \\
	&&\leq \sum_i K(u_i, x_i, y_i | \zetZero) + \sum_i K(v_i | \zetThree) + O(n) \\
	&&\leq \sum_i K(u_i, x_i | \zetZero) + \sum_i K(y_i, v_i | \zetTwo) - \sum_i K(y_i| \zetTwo) + O(n \log \log m)\\
	&&\leq K(T_{ux}) + K(T_{yv}) - O(n\log m).
\end{eqnarray*} 

\noindent
Let $\la.,.\ra$ be a computable bijective pairing function such that 
for all $a,b \in \omega$, $l(\overline{\la a,b \ra}) \leq l(\overline{a}) + l(\overline{b})$.  Let:
\begin{eqnarray*}
 x'_i &=& \la u_i, x_i \ra \\
 y'_i &=& \la y_i, v_i \ra.
\end{eqnarray*}

\noindent
To finish the proof it suffices to show that 
\begin{equation} \label{eq:xyLengthA}
l(\overline{x'}) + l(\overline{y'}) \leq N \leq O(nm).
\end{equation}
Remark that because $x_i \in 2^m$, $l(\overline{x_i}) \leq 2m$ 
and because $y_i = K(T_{x,i})$, $l(\overline{y_i}) \leq^+ 2\log m$:
\begin{eqnarray*}
 l(\overline{x'_i}) + l(\overline{y'_i}) &\leq& l(\overline{u_i}) 
			+ l(\overline{x_i}) + l(\overline{y_i}) + l(\overline{v_i})  \\
		&\leq& l(\overline{D_{ux,i}}) + 2m + 2\log m + l(\overline{D_{yv,i}}). \nonumber
\end{eqnarray*}
Choose $N=3mn$. For $X = ux,yv$, $\sum_i D_{X,i} \leq N+1$, 
and therefore $\sum_i l(\overline{D_{X,i}}) \leq 3n \log m$.
This shows that for $m$ large enough: 
\begin{equation} \label{eq:xyLengthB}
 l(\overline{x'}) + l(\overline{y'}) \leq 3mn = N.
\end{equation}
This shows equation \eqref{eq:xyLengthA}.
\end{proof}

\begin{corollary}
 For some $c>0$, for all but finitely many $n$, there exist a $z \in 2^{2n}$ such that:
 \begin{equation} \label{eq:decisionDiffProp}
  K(0 \rar z_1;  ... ; z_{2n-2} \rar z_{2n-1})
   + K(z_1 \rar z_2; ... ; z_{2n-1} \rar z_{2n}) - K(z) \geq cn. 
 \end{equation}
\end{corollary}

\begin{proof}
 Let $x',y'$ be as constructed in the proof. 
 Let $\ovl{x'_i}$ and $\ovl{y'_i}$ be binary prefix-free encodings
 corresponding to the definition of $l(\overline{x})$.
of $x'_i$ and $y'_i$, $i\leq n$.   Define $z$:
 \begin{eqnarray*}
 z =& \ovl{x'}_{1,1}, 0, ... ,\ovl{x'}_{1,l(\ovl{x'_1})}, 0, \\
   &  0, \ovl{y'}_{1,1}, ... ,0, \ovl{y'}_{1,l(\ovl{y'_1})}, \\
   & ... \\
   & \ovl{x'}_{n,1}, 0, ... ,\ovl{x'}_{n,l(\ovl{x'_n})}, 0, \\
   &  0, \ovl{y'}_{n,1}, ... ,0, \ovl{y'}_{n,l(\ovl{y'_n})}.
 \end{eqnarray*}
 Since $\sum_{i\leq n} l(\overline{x'_i}) + l(\overline{y'_i}) \leq 3mn$, 
we have that $z \in 2^{\leq 6n}$.
This shows that for all but finitely many $n$ a string of length 
maximally $6mn$ exists that satisfies the inequality of the lemma. 
By appending zeros to the end of $x'$ and $y'$, equality 
\eqref{eq:decisionDiffProp} can be satisfied for every $n$.
\end{proof}

\begin{acknowledgement}
The author is grateful for the comments of A. Shen on early proof attempts and motivation 
to write out a full exact proof.
\end{acknowledgement}

\bibliographystyle{plain}
\bibliography{../bib/practCausalities,../bib/statisticalCausalities,../bib/eigen,../bib/bib}

\end{document}